\documentclass[journal,11pt,draftclsnofoot,onecolumn]{IEEEtran}
\usepackage{etex}

\usepackage{lipsum}
\usepackage{citesort}
\usepackage{varwidth}
\usepackage{dsfont}
\usepackage{algpseudocode}
\usepackage{algorithm}
\usepackage{amsthm}
\usepackage{amssymb}
\usepackage{arydshln}
\usepackage{enumerate}
\usepackage{amssymb}
\usepackage[final]{graphicx}
\usepackage{float}
\usepackage{epstopdf}
\usepackage{pbox}
\usepackage{tikz}
\usetikzlibrary{arrows,matrix,positioning}
\usepackage{ctable}
\usepackage{threeparttable}
\usepackage[section,subsection,subsubsection]{extraplaceins}
\usepackage{multirow}
\usepackage{subfigure}
\usepackage{booktabs}
\usepackage{amsmath}
\usepackage{bm}
\usepackage{tikz}
\usetikzlibrary{arrows}
\tikzstyle{int}=[draw, fill=blue!20, minimum size=2em]
\tikzstyle{init} = [pin edge={to-,thin,black}]

\newcommand{\bb}[1]{\textcolor{black}{#1}}

\DeclareMathOperator*{\argmin}{arg\,min}
\DeclareMathOperator*{\argmax}{arg\,max}

\usepackage{multicol}
\newcommand{\ra}[1]{\renewcommand{\arraystretch}{#1}}
\newcommand{\RNum}[1]{\uppercase\expandafter{\romannumeral #1\relax}}
\newtheorem{theorem}{Theorem}
\newtheorem{lemma}{Lemma}

\theoremstyle{stylename}
\newtheorem*{definition}{Definition}

\usepackage{tkz-berge}

\usetikzlibrary{positioning,chains,fit,shapes,calc, arrows, shadows,trees}
\usepackage{calc}
\usepackage{multicol}
\usepackage{lipsum}
\usetikzlibrary{decorations.markings, arrows, automata}
\tikzstyle{vertex}=[circle, draw, inner sep=0pt, minimum size=5pt]
\tikzset{symbol/.style={rectangle, draw, very thick,
minimum size=10mm, rounded corners=1mm}}
\tikzset{symbol2/.style={rectangle , draw,  thick,
minimum size=35mm, rounded corners=1mm}}
\newcommand{\vertex}{\node[vertex]}

\newcolumntype{C}[1]{>{\centering\arraybackslash}p{#1}}

\usepackage{xparse}
\usepackage{tikz}
\usetikzlibrary{matrix,backgrounds}
\pgfdeclarelayer{myback}
\pgfsetlayers{myback,background,main}

\tikzset{mycolor/.style = {line width=1bp,color=#1}}%
\tikzset{myfillcolor/.style = {draw,fill=#1}}%

\NewDocumentCommand{\highlight}{O{blue!40} m m}{%
\draw[mycolor=#1] (#2.north west)rectangle (#3.south east);
}

\NewDocumentCommand{\fhighlight}{O{blue!40} m m}{%
\draw[myfillcolor=#1] (#2.north west)rectangle (#3.south east);
}
\begin{document}
\title{\centering \bb{Fountain Codes with Nonuniform Selection Distributions through Feedback}}

\renewcommand{\footnoterule}{%
  \kern -3pt
  \hrule width \textwidth height 0pt
  \kern 3pt
}


\author{\IEEEauthorblockN{Morteza Hashemi\IEEEauthorrefmark{1}, Yuval Cassuto\IEEEauthorrefmark{2}, and Ari Trachtenberg\IEEEauthorrefmark{1} \\}
\IEEEauthorblockA{\IEEEauthorrefmark{1}Dept. of Electrical and Computer Engineering, Boston University \\}
\IEEEauthorblockA{\IEEEauthorrefmark{2}Dept. of Electrical Engineering, Technion, Israel Institute of Technology \vspace{-1cm}}
\thanks{This paper was presented in part at the 51st and 52nd Annual Allerton Conference on Communication, Control, and Computing \cite{allerton2013,allerton2014}.}}



\maketitle

\begin{abstract}
One key requirement for fountain (rateless) coding schemes is to achieve a high \emph{intermediate} symbol recovery rate. Recent coding schemes have incorporated the use of a feedback channel to improve intermediate performance of traditional rateless codes; however, these codes with feedback are designed based on \emph{uniformly at random} selection of input symbols. In this paper, on the other hand, we develop feedback-based fountain codes with dynamically-adjusted nonuniform symbol selection distributions, and show that this characteristic can enhance the intermediate decoding rate.   We provide an analysis of our codes, including bounds on computational complexity and failure probability for a maximum likelihood decoder; the latter are tighter than bounds known for classical rateless codes. Through numerical simulations, we also show that feedback information paired with a nonuniform selection distribution can highly improve the symbol recovery rate, and that the amount of feedback sent can be tuned to the specific transmission properties of a given feedback channel.  

\end{abstract}


\section{Introduction}
Reliable communication over erasure channels has emerged as a key technology for various networked applications, 
for example digital video broadcasting and over-the-air software updates.  In applications where there exists a high-throughput feedback channel, automatic repeat request (ARQ) protocols guarantee reliability over erasure channels.  However, when such feedback channels are not available, rateless codes, such as the capacity achieving Luby-Transform (LT)~\cite{luby2002lt} and Raptor codes~\cite{shokrollahi2006raptor}, can often provide reliable communication for sufficiently long block lengths. These codes have a well-known all-or-nothing decoding property (the so-called ``waterfall'' phenomenon), where a jump in the fraction of decoded input symbols occurs near the very end of the decoding process. For applications with real-time requirements, however, it is desirable to be able to recover symbols as decoding proceeds, i.e., to achieve a high intermediate symbol recovery rate.     

In fact, the intermediate performance of classical codes can be improved by incorporating the use of a feedback channel.  For instance, a decoder in Real-Time (RT) oblivious~\cite{beimel2007rt} and Shifted-LT (SLT)~\cite{hagedorn2009rateless} codes sends the number of recovered symbols back to the transmitter, and this feedback is used to modify the degree distribution at the encoder.  
Previous feedback-based rateless codes are mostly based on adjusting the degree of encoding symbols, e.g., by shifting the degree distribution in the SLT codes. However, after a degree $d$ is picked for an encoding symbol, $d$ input symbols are chosen \emph{uniformly at random} and xored to form the symbol. Moreover, the encoder does not have full freedom in controlling the number of feedbacks transmitted.  
 
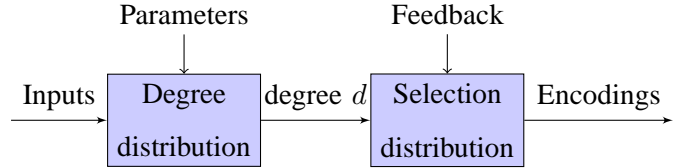
\begin{figure}
\centering
\begin{tikzpicture}[node distance=3.5cm,auto,>=latex']
    \node [int, pin={[init]above:Parameters}, align=center] (a) {Degree \\ distribution};
    \node (b) [left of=a,node distance=2.3cm, coordinate] {a};
    \node [int, pin={[init]above:Feedback}, align=center] (c) [right of=a] {Selection \\ distribution};
    \node [coordinate] (end) [right of=c, node distance=3cm]{};
    \path[->] (b) edge node {Inputs} (a);
    \path[->] (a) edge node {degree $d$} (c);
    \draw[->] (c) edge node {Encodings} (end) ;
\end{tikzpicture}
\caption{Two-step rateless encoder with a degree distribution and nonuniform symbol selection distribution.}
\label{two-step}
\end{figure}

In this paper, we develop a class of rateless coding schemes that optimize for high intermediate symbol recovery rate.  At its core, our encoder uses a nonuniform selection distribution that is dynamically adjusted based on feedback information. Fig. \ref{two-step} depicts a schematic of our two-step encoder, where we illustrate
that the inputs are chosen according to a feedback-based selection distribution, rather than uniformly at random. Feedback messages contain information on the distance between a received encoding symbol and the set of already decoded symbols at the receiver.  In the \emph{general form} of our codes, the encoder estimates the probability that each input symbol has been decoded (at the decoder), and these estimates are then used to dynamically tune the selection of input symbols within subsequent transmissions. This method enables the encoder to naturally track the decoding progress and generate encoding symbols that result in a faster decoding rate compared with a uniform selection of input symbols. This class of codes is suitable for the scenarios with relatively large feedback budgets, although we allow the decoder to specific control \emph{when} feedback occurs (according to the budget).  

On the other hand, the \emph{primitive form} of our code is designed based on a parsimonious use of the feedback channel. In this case, the encoder learns which symbols have been decoded, and those symbols will be assigned with a selection probability of zero for subsequent encodings. This coding scheme is suitable for applications with limited feedback capacity such as satellite networks \cite{byers2002digital}, as we require the decoder to opportunistically send just one bit of feedback when certain conditions are met. Note that the coding schemes proposed in this work are presented as enhancements of LT codes for the case some feedback communication is available. The motive to base our codes on the LT degree distributions is to accommodate cases where feedback is extremely limited or completely unavailable, in which case we fall back to the standard LT performance. That said, the same methodology can apply to different rateless codes in the literature, and to others to be proposed in the future.
   
\subsection{Organization}
The rest of this paper is organized as follows.  In Section~\ref{related work} we describe the problem setup and review various related coding schemes. Section \ref{RLT} presents the most general form of our coding scheme. Section~\ref{Delete-and-Conquer} describes the primitive form of our codes adapted for constrained feedback applications.  Coding analysis for short block lengths is presented in Section~\ref{overhead}, followed by maximum likelihood decoder analysis in Section \ref{maximum-likelihood}. Simulation results are presented in Section~\ref{simulation-results}.  We conclude with overall thoughts in Section \ref{conclusion}.


\section{Background}
\label{related work}
This section describes the problem setup and some previous work on rateless coding. 
\subsection{Preliminaries}
\label{preliminaries}

In the rateless coding setup, it is assumed that an encoder (broadcaster) has $k$ input symbols to transmit to all receivers over an erasure channel, and that there may exist a feedback channel through which receivers can send some information back to the encoder.   Luby Transform (LT) rateless codes  \cite{luby2002lt}, as the first practical realization of fountain codes, support full recovery of $k$ input symbols using an expected number of $k+O\left(\sqrt{k} \ln ^2(k/\delta)\right)$ error-free transmissions with a given recovery failure probability $\delta$.  To generate an output symbol, the encoder first picks a coding degree $d$ according to the Robust Soliton distribution \cite{luby2002lt}. Next, $d$ input symbols are chosen uniformly at random without replacement, and their sum over an appropriate finite field forms the output symbol. Indices of the $d$ selected input symbols, referred to as \emph{neighbors} of the output symbol, are made available (i.e., as meta-information) to the decoder. In total, the coding operations incur the  computational cost of $O\left(k \ln (k/\delta)\right)$.

The LT decoder (so-called Peeling decoder) uses a simple message passing algorithm, with a complexity typically less than traditional Gaussian elimination methods.  In one variant, the decoder finds all encoding symbols with degree $1$, whose neighbor can be immediately recovered. These recovered input symbols are then excluded from all output symbols that have them as neighbors, reducing the number of unknowns in those encoding symbols by one. This process continues until there exists no encoding symbol with degree $1$. Decoding succeeds if all input symbols are recovered; alternatively, decoding fails if, at some point, there is no output symbol with degree $1$.


\subsection{Related work}
Both fixed rate low-density parity-check (LDPC) codes~\cite{gallager1960low} and Turbo codes~\cite{berrou1993near} are capable of correcting bit errors, as well as erasures.  Byers \emph{et al.} in \cite{byers1998digital} have presented fixed rate Tornado codes as a class of simplified capacity-achieving LDPC codes. Within the context of rateless coding, random linear codes (see, for example, \cite{mackay2005fountain}) are well known due to their low communication overhead, but the encoding and decoding computations make them practical only for small message sizes. On the other hand, Luby Transform (LT)~\cite{luby2002lt} codes and their extensions such as Raptor codes \cite{shokrollahi2006raptor} are examples of rateless codes that are asymptotically optimal and also have computationally efficient encoding and decoding algorithms; unfortunately, they usually have poor performance for small block sizes~\cite{sorensen2012rateless} and various optimization methods~(e.g., \cite{hyytia2007optimal}) have been proposed for these cases. 

In some applications, like video streaming, intermediate symbol recovery is important, as it is desirable to decode some symbols before an entire frame has been received. The authors in~\cite{talari2009rateless} design a degree distribution for high intermediate symbol recovery rates.  Recently, there have also been proposed rateless protocols that utilize side information fed back from the decoder to the encoder. Based on the type of feedback used, they can be divided in the following categories:
\begin{itemize}
\item the receiver sends the \emph{number} of decoded symbols to the transmitter;
\item the receiver suggests to the transmitter \emph{what kind of degrees} it should use for future encodings; or
\item the receiver notifies the transmitter of \emph{which} input symbols have been recovered.
\end{itemize}

In the Real-Time (RT) oblivious codes~\cite{beimel2007rt}, the encoder starts with degree one symbols, and it increments the degree of encoding symbols based on feedback messages. In this case, feedbacks contain information on the number of recovered symbols.  Shifted LT (SLT) codes proposed in~\cite{hagedorn2009rateless} use the same type of feedback information as the RT codes, but instead of explicitly increasing the encoded symbols degree, the encoder shifts the Robust Soliton degree distribution.  There also exist rateless-type codes with real-time properties that allow intermediate knowledge of some input symbols as the decoding progresses. The authors in~\cite{kamra2006growth} propose Growth codes for the data collection within lossy sensor networks. Similar to the RT and SLT codes, Growth codes' degree increases as the coding progresses.

As another type of feedback, the receiver in~\cite{cassuto2011line}  has the ability to control the decoding progress by requesting particular degrees. In this method, the average number of output symbols required for decoding $k$ input symbols is shown to be upper bounded by $1.236k$.  Yet another type of feedback in \cite{sorensen2012rateless} contains the identity of recovered symbols, which are used by the encoder to redesign the degree distribution for subsequent transmissions. Recently, the authors in \cite{talari2014robust} have proposed a heuristic to use a hybrid feedback-based rateless codes, called LT-AF, in which the receiver alternates between two types of feedback messages: the first type of feedback contains the number of decoded symbols as in the SLT and RT codes, while the receiver requests a specific  input symbol through the second type of feedback.

In this paper, the type of feedback used is based on distance information by which the encoder learns about the state of individual symbols at the decoder side. Based on the feedback information, the encoder tunes a nonuniform selection distribution to choose neighbors of encoding symbols. This is a key point as all previous rateless codes are built upon a uniformly at random selection of neighbors. Moreover, we do not assume that the feedback channel is high bandwidth; instead, we strive for a parsimonious use of the feedback channel.  Indeed, in the primitive form of our codes, the feedback is exactly one bit (plus some header information) for each of a small fraction of received symbols.

\section{Nonuniform Rateless Codes: General Form}
\label{RLT}
Previous rateless codes with feedback are mostly designed based on modifying the output symbols degree distribution according to feedback information, e.g., by shifting the degree distribution, or by explicitly  increasing the degree. In these schemes, when a coding degree $d$ is picked,  $d$ input symbols are selected uniformly at random to construct an encoding symbol. Moreover, in most of previous works, feedback information does not provide a complete picture of the decoding state at the receiver side. For instance, sending the number of recovered symbols in~\cite{beimel2007rt,hagedorn2009rateless} does not provide information about the decoded symbols themselves. Within this context, we present a nonuniform rateless coding scheme wherein various input symbols are selected based upon a nonuniform distribution. In particular, the selection distribution is tuned according to feedback messages, which contain the distance of received symbols to the set of already recovered symbols at the receiver. The definition of distance quantity is as follows: 

\begin{definition}
Given a set of recovered symbols $C$ and an encoding symbol $y$ that has a set of neighbors $A$,  the \textbf{distance} between $y$ and $\mathcal{C}$ is defined as:
$$
\text{dist}(y, C) = \sum_{x_{i} \in A} \large{\mathds{1}}_{x_{i} \notin C},
$$
where $\large{\mathds{1}}_x$ is an indicator function that is equal to 1 if and only if $x$ is true.
\end{definition}

The distance quantity simply corresponds to the number of neighbors of $y$ that are not already decoded.
As an example, suppose that input symbols $x_1$, .., $x_4$ are encoded and transmitted in the following order: $y_1 = x_1 + x_2$, $y_2 = x_1 + x_4$, $y_3 = x_4$, and $y_4 = x_1 + x_2$\footnote{For the sake of clarity,  we assume that encoding and decoding are performed over the field $\mathbb{F}_2$.}. The distance from either $y_1$ or $y_2$ to the set of recovered symbols is $2$; thereafter, from $y_3$ the distance is $1$, and finally, from $y_4$ the distance is $0$ (as $x_1$ and $x_2$ will be decoded after receiving $y_3$).

Ultimately, the goal of the encoder is to generate encoding symbols based on the state of the decoder in such a way that more ``helpful'' symbols have a higher selection probability. To this end, the encoder uses distance information to estimate the probability that each input symbol has been decoded (at the receiver), and these estimates are used to bias the selection of input symbols. 
In this approach, the receiver can adjust the number of feedbacks using a parameter $s$ so that one feedback transmission follows after every $s$ received encoding symbols.  The parameter $s$ can be set to any arbitrary value, depending on the feedback channel available.

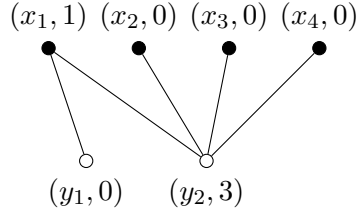
\begin{figure}[t]
\centering
  \begin{tikzpicture}[scale=1]
    \vertex[fill] (x1) at (0,1.5) [label=above:{$(x_{1}, 1)$}]{};
	\vertex[fill] (x2) at (1.2,1.5) [label=above:{$(x_{2}, 0)$}]{};
	\vertex[fill] (x3) at (2.4,1.5) [label=above:{$(x_{3}, 0)$}]{};
	\vertex[fill] (x4) at (3.6,1.5) [label=above:{$(x_{4}, 0)$}]{};
	\vertex (y1) at (.5,0) [label=below:{$(y_{1},0)$}]{};
	\vertex (y3) at (2.1,0) [label=below:{$(y_{2},3)$}]{};
	\path
		(x1) edge node[left]{} (y1)  
		(x1) edge node[left]{} (y3)  
		(x4) edge node[left]{} (y3)
		(x2) edge node[left]{} (y3) 
		(x3) edge node[left]{} (y3);
    \end{tikzpicture} 
    \label{fig:distance-graph-case2}

\caption{Distance graph labeling: A label $(x_i, q_i)$ implies that the input symbol $x_i$ has been decoded with probability $q_i$ up to the current state. Labels of output nodes $y_i$ are defined to be the number of neighbors of $y_i$ with a label of less than $1$.}
\label{distance-graph}
\end{figure}

\subsection{Processing distance information} 
In order to process distance information, the encoder constructs a bipartite graph wherein input symbols are placed on the top and encoding symbols at the bottom, as shown in Fig. \ref{distance-graph}. 
In this graph, \emph{labels} are assigned to input and output symbols. In particular, label of an input symbol corresponds to its probability of having been decoded, while label of an output symbol $y$ represents the number of neighbors of $y$ with label less than $1$. For instance, assume that after $t$ feedbacks, $n_t$  neighbors of $y$ are labeled $1$ (i.e, they have been decoded). Therefore, the label of $y$, denoted by $l_t$, is calculated as 
\begin{equation}
l_t = d - n_t;
\label{y-label}
\end{equation}
where $d$ is the degree of $y$. In this equation, the encoder excludes the recovered neighbors from the labeling process. Next, in order to calculate the label of an input symbol, we assume that the $t$-th feedback message contains the distance $f_t$ corresponding to the encoding symbol $y=\sum_{j \in A} x_j$. The label of a constituent symbol $x_j$ is then defined as: 

\begin{align}
\bb{q_{j,t} = \text{max}  \left\{q_{j,t-1}, \frac{{{l_t-1} \choose f_t}}{{l_t\choose f_t }} \right\} =  \text{max}  \left\{q_{j,t-1}, \frac{l_t-f_t}{l_t} \right\}. } 
\label{labeling}
\end{align} 
It should be noted that $l_t$ is the number of neighbors with a label less than $1$ and $f_t$ is the number of undecoded neighbors.  Therefore, the probability of having the neighbor $j$ decoded is calculated as $\frac{l_t-f_t}{l_t}$. Finally, after receiving a new feedback message, $q_j$ is updated to the maximum of its previous value and the calculated probability at the current step.  
For instance, assume that the encoding symbol $y=x_1+x_2+x_3+x_4$ has a distance of $2$ with the current state of decoder, meaning that two neighbors of $y$ have not been decoded yet (the encoder does not know which two symbols). If the encoder has already assigned label  $1$ to $x_1$ (i.e., $x_1$ has been decoded), then the encoder uniformly divides the distance of $2$ between the remaining symbols (i.e., $x_2$, $x_3$, and $x_4$), suggesting that each of them has been decoded with probability  $\frac{3-2}{3} = \frac{1}{3}$. It should be noted that the subscript $t$ in $q_{j,t}$ represents the evolution of $q$ as the coding proceeds. For simplicity, we drop it in our discussion.

Our labeling process tracks the state of the decoder by answering this question: \emph{what is the probability that an individual symbol $x_j$ has been decoded up to this point?} As an example, Fig. \ref{decoding-probability} shows a realization of input symbols at the encoder, where input symbols are assigned with a probability of having been decoded. In this case, $q_j=1$ (white color) implies that symbol $j$ has been recovered, while $q_j=0$ (black color) means that symbol $j$ has not been recovered yet. \bb{Therefore, input symbols are assigned with a weight between $0$ and $1$, which is used in the selection distributions defined in Section \ref{nonuniform-distribution}.}
\begin{figure}[t]
\centering
\includegraphics[scale=.4, trim=1cm 2cm 1.5cm 1.7cm, clip]{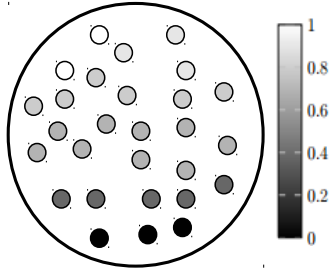}
\caption{The encoder estimates the probability of having been decoded (i.e., $q_j$'s) for input symbols. Probability $1$ (white color) implies that the symbol has been decoded, while probability  $0$ (black color) shows that the symbol has not been decoded yet.}
\label{decoding-probability}
\end{figure}

To examine the accuracy of  the estimated probability values against the actual decoder state, we use the Mean Absolute Error (MAE) quantity.  Assume that $b_j$ is an indicator function representing the state of symbol $j$ at the decoder such that $b_j = 1$ if symbol $j$ has been decoded and $b_j = 0$ otherwise. MAE is then calculated as:
$$
\text{MAE} = \frac{1}{k} \sum_{j=1} ^k |q_j - b_j|,
$$ 
in which $q_j$'s are estimated using \eqref{labeling}. Fig. \ref{MAE} shows the MAE quantity averaged over all feedback messages transmitted as the interval of feedback transmission (i.e., parameter $s$) increases. The results illustrate that decreasing the interval of feedback transmission (i.e., higher feedback rate) decreases the estimation error.

\emph{Remark 1 (cumulative feedback information):} Distance messages accumulate information across all received feedbacks. Specifically, assume that there exist $k$ input symbols at the encoder, and after receiving a new feedback message, say the $t$-th feedback, the encoder updates the probability vector $\mathbf{q}_t=\left(q_{1,t}, q_{2,t}, ..., q_{k,t}\right)$, where $q_{j,t}$ is the probability that the input symbol $j$ has been decoded. The encoder updates probability values corresponding to neighbors of the encoding symbol, and other probability values remain unchanged. This update mechanism allows the encoder to accumulate information across all feedback messages, noting that in previous feedback-based schemes (e.g., sending number of recovered symbols in \cite{beimel2007rt,hagedorn2009rateless}), a new feedback makes previously received feedback information obsolete.   

\begin{figure}[t!]
\centering
\includegraphics[scale=.3, trim=1cm .5cm 1cm 1cm, clip]{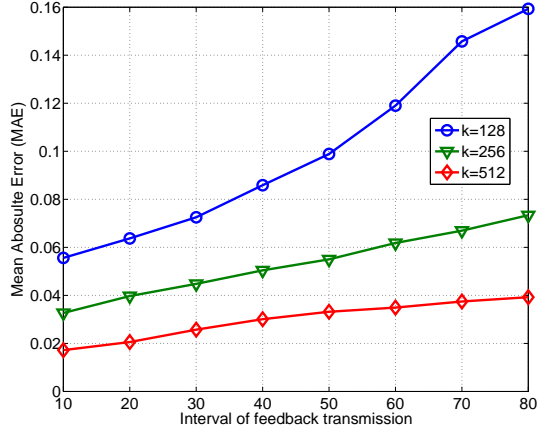}
\caption{Mean Absolute Error (MAE) for the estimated probability of having been decoded as the interval of feedback transmission (parameter $s$) increases. }
\label{MAE}
\end{figure}  

Distance feedbacks provide implicit information about the decoder's graph; however, it should be noted that one can envision other techniques to learn about the decoding graph. For instance, the decoder can send the whole decoding graph back to the encoder, and thus, the encoder would have full knowledge about the state of input symbols at the decoder's side. This method, however, incurs high communication overhead on the back channel. Next, we define nonuniform symbol selection distributions based on probability values $q_j$'s. 

\subsection{Nonuniform symbol selection}
\label{nonuniform-distribution}
We discussed that the encoder uses distance information to learn about the state of decoder. For the sake of concreteness, assume that the encoder estimates the input symbol $x_j$ has been decoded with probability  $q_j$, and it has  probability $p_j$ to be included in the next encoding symbol with degree $d$. We aim to design a symbol selection distribution that picks those $d$ input symbols that can achieve a maximum decoding progress. Specifically, the selection distribution should select $d-1$ symbols which have been recovered with a high probability, and a single symbol that has not been recovered with a high probability. To put in a formal framework, we have the following definition. 

\begin{definition}
For a given input symbol $x$ and a set of $d-1$ input symbols $A$, the Decoding Probability function $DP(x, A)$ is defined as the probability of immediate decoding the input symbol $x$ after receiving $y=x + \sum _{i \in A} x_i $.  
\end{definition}

 In order to decode an input symbol $x_j$ within a transmission, the transmitted symbol with the degree of $d$ should include the undecoded symbol $x_j$ and $d-1$ already decoded symbols. Symbol $x_j$ is not decoded with probability $1-q_j$, and $d-1$ symbols belonging to the set $A$ have already been decoded with probability $\prod_{i \in A} q_{i}$.   Therefore, at each step to transmit a symbol of degree $d$, the encoder should choose $d$ input symbols $\left( x_j^*, A^* \right) $ satisfying:
 \small
\begin{align}
 \left( x_j^*, A^* \right)&  = \argmax_{(x_j,A)} \ DP (x_j, A)  =   \argmax_{(x_j,A)}  \left(1-q_{j}\right) \prod_{\substack{i \in A \\ i \neq j}} q_{i}.
 \label{optimize-code}
\end{align}  
\normalsize
\bb{The solution of this optimization problem is deterministic; in other words, the encoder always picks $d-1$ symbols with the largest value of $q$ xored with a single symbol with the smallest value of $q$. However, it is desirable to preserve the same behavior with a probabilistic scheme such that if an input symbol $j$ is included in the solution of the deterministic formulation, it would also have a high probability to be picked by the probabilistic method.  This results in the following scheme to define the selection probability $p_j$:}
$$
  p_j \propto \left\{
  \begin{array}{l l}
    1-q_j &  \text{if} \ \ 0 \leqslant q_j < \frac{1}{2}; \\
    q_j &  \text{otherwise}.
  \end{array} \right.
$$

In the second step of designing the selection distribution, we note that a single unrecovered (with high probability) symbol should be included within each encoding symbol. Therefore, based on the value of $q_j$'s, input symbols are divided into two subsets: $U$ containing undecoded symbols, and $D$ containing decoded symbols.
Input symbols with $0 \leqslant q < \frac{1}{2}$ are included in $U$ and the rest are added to the set $D$, and thus we may construct an encoding symbol of degree $d$ by selecting one symbol from $U$ based on the selection distribution $P_U$, and $d-1$ symbols from $D$ according to the distribution $P_D$. Selection distributions $P_U$ and $P_D$ are defined as follows:

 \begin{equation}
  P_U (j) = \left\{
  \begin{array}{l l}
    \frac{1-q_j}{\sum_{i=1}^{k-m} 1-q_i} &  \text{if} \ j \in U ; \\
    0 &  \text{otherwise}.
  \end{array} \right.
  \qquad
   P_D (j) = \left\{
  \begin{array}{l l}
    \frac{q_j}{\sum_{i=1}^{m} q_i} &  \text{if} \ j \in D ; \\
    0 &  \text{otherwise}.
  \end{array} \right.
    \label{nonuniform}
 \end{equation}
\bb{In the distributions, $m$ is the size of subset $D$. Finally, the encoder transmits the xor of $d$ selected symbols.}  

This scheme based on the distributions $P_D$ and $P_U$ is refereed to as the \emph{All-Distance} codes since all distance feedbacks are needed to estimate probability value $q_j$'s. Next, we relax this scheme in a way that, instead of sending all distance values, the decoder \emph{quantizes} distance values and allocates only a single bit feedback for each received encoding symbol.

\subsection{Quantized distance codes}
The All-Distance codes work based on estimating probability values $q_j$'s from distance information. In  this case, at most $n\log(d_{\text{max}})$ bits are sent back to the encoder, as each of $n$ encoding symbols can have distance $d_{max}$, which is the maximum degree of an encoding symbol, noting that $d_{max}$ can be at most equal to $k$. 

To limit amount of feedback, we consider a scheme with a single bit feedback per received encoding symbol. In particular, this scheme is based on the same idea of splitting input symbols into two subsets; however, instead of having an exact estimation of probability value $q_j$'s, the decoder decides to send a feedback $0$ or $1$ based on the distance of a received symbol. More precisely, the decoder calculates the ratio of distance to degree for a received symbol, and if the ratio is larger than $\frac{1}{2}$, it implies that majority of neighbors within the received encoding symbol have not been recovered. In this case, the decoder allocates a single bit of $0$ as the feedback message. On the other hand, if the calculated ratio is smaller than $\frac{1}{2}$, it shows that majority of neighbors have been decoded and feedback message would be $1$. \bb{To limit the number of feedback transmissions, the receiver bundles the 1-bit feedback messages together for every interval of $s$ received encoding symbols, and sends the $s$-bit messages back to the encoder.  }

At the encoder side and upon receiving a feedback message $0$, corresponding neighbors are assigned with $q_j = 0$ and thus added to the subset $U$. Conversely, if the received feedback contains a bit of $1$, corresponding neighbors are assigned with $q_j = 1$ and grouped into $D$. This \emph{quantized} version of $q_j$ is equivalent to evaluating $\lfloor q_j \rceil$ in \eqref{nonuniform} ($\lfloor x \rceil$ rounds $x$ to its nearest integer). As a result, the $P_U$ and $P_D$ distributions would become uniformly distributed over the subsets $D$ and $U$ respectively. However, it should be noted that with a high probability only a single undecoded symbol is included within each transmission. Hence, splitting a single uniform distribution defined over all input symbols (as it has been used in previous rateless codes) into two disjoint uniform selection distributions, can significantly improve the intermediate performance of rateless codes.  

In terms of total amount of feedback, decoder sends exactly one bit feedback per received encoding symbol, where the total number of encoding symbols is $(1+\epsilon)k$ for a small value of $\epsilon$. Recall that the motivation behind the distance type feedback is to learn about the state of individual symbols at the decoder side. However, an alternative and trivial solution includes sending the identity of recovered symbols back to the encoder with potentially more feedback that could be up to $k\log(k)$ bits. \bb{In this case, it may not be clear how the encoder uses deterministic information on the identity of recovered symbols. In fact, the authors in \cite{sorensen2012rateless} use the identity of recovered symbols in order to redesign the primary degree distribution through a computationally expensive algorithm; on the other hand, we use distance information through a probabilistic scheme to dynamically assign nonuniform selection weights to input symbols.}


\section{Nonuniform Rateless Codes: Primitive Form}
\label{Delete-and-Conquer}
In most of communication systems, a nominal utilization of the back channel is desirable as the bandwidth is mainly provisioned for forward transmissions. In the previous section, we presented a nonuniform coding scheme based on distance feedbacks, wherein all distance information are fed back to the encoder. In the scheme based on quantized distance information, decoder needs one bit feedback per received encoding symbol. In this section, we establish a coding scheme called Delete-and-Conquer with a more limited use of the feedback channel. In this case, the decoder is allowed to transmit one bit feedback for  a small fraction of received encoding symbols, when certain conditions are met.

\subsection{Delete-and-Conquer codes}
Recalling the definition of the distance metric, a distance $0$ happens if and only if all neighbors of the received encoding symbol have already been decoded. Similarly, a distance $1$ occurs in the case that there is only a single undecoded neighbor, which can then be recovered uniquely. In other words, a distance of 0 or 1 provides information about the recovery of neighbors that are part of a received linear combination.

A Delete-and-Conquer encoder performs similar to the LT encoder in that it first picks a coding degree $d$ from the degree distribution. However, in the second step the encoder selects $d$ symbols from a \emph{subset} of input symbols. Specifically, upon receiving a feedback message, the encoder assigns a selection probability of zero to the  neighbors of the acknowledged encoding symbol, while remaining symbols would have an equal selection probability. Intuitively, the encoder deletes recovered symbols and continues with a smaller block of symbols; in so doing, the encoder also rescales the primary degree distribution (e.g., the Robust Soliton distribution denoted by $\Omega_k$) to the smaller set of input symbols with size $k-m$, in which $m$ is the number of deleted symbols. Excluding recovered symbols from future transmissions reduces the computational complexity at the encoder and decoder. Algorithm 1 gives the pseudo-code of the Delete-and-Conquer encoding scheme.

\begin{algorithm}[t]
\caption{Delete-and-Conquer Encoding $(x_1, x_2, .., x_k)$}
\begin{varwidth}{\dimexpr\linewidth-2\fboxsep-2\fboxrule\relax}
\begin{algorithmic}[1]
\State $z \gets 0$ and $m \gets 0$ 
\State $\mathcal{A} \gets \{x_1, x_2, ..., x_k\}$ and $\mathcal{B} \gets 	\emptyset$
\While{$ z < k $}
\State Pick a coding degree $d$ from the distribution $\Omega_{k-m}$ 
\State Select $d$ symbols uniformly at random from set $\mathcal{A}$
\State Send symbol $y$ as XOR of $d$ selected symbols
\If {feedback$(y)$ = \textbf{true}}
\State $\mathcal{C} \gets$ Neighbors of $y$
\State $\mathcal{B} \gets B \cup \mathcal{C}$ 
\State $m \gets |\mathcal{B}|$
\State $\mathcal{A} \gets \mathcal{A} \setminus \mathcal{B}$
\EndIf
\If {Terminate = \textbf{true}}
\State $z=k$
\EndIf
\EndWhile
\end{algorithmic}
\end{varwidth}%
\end{algorithm}

The Delete-and-Conquer decoder is based on Peeling decoder with a slight modification that upon receiving a new encoding symbol, the decoder checks if the distance is equal to $0$ or $1$. The $0$ and $1$ distance feedbacks are indeed a generalization of the traditional acknowledgment to the coded cases in that they notify the recovery of a group of input symbols involved in an encoding. The pseudo-code of the Delete-and-Conquer decoding is provided in Algorithm 2. 

\begin{algorithm}[h]
\caption{Delete-and-Conquer Decoding of $k$ symbols}
\begin{varwidth}{\dimexpr\linewidth-2\fboxsep-2\fboxrule\relax}
\begin{algorithmic}[1]
\State $\mathcal{S} \gets \emptyset$   \Comment{$\mathcal{S}$ is the set of recovered symbols}
\While{$|\mathcal{S}| < k$}
\State $y \gets $ Received encoded symbol
\If { Distance$(y, \mathcal{S}) = 0$ or $1$}
\State Send a feedback and set feedback$(y)$ \textbf{true}
\EndIf
\State \textbf{call} Peeling-Decoder
\State Update $\mathcal{S}$

\If {$|\mathcal{S}| = k$}
\State Terminate = \textbf{true}
\EndIf
\EndWhile
\State
\Function{Distance}{$y$, $\mathcal{S}$}
\State distance $\gets 0$ 
\For{all neighbors $x_i$ of $y$}
\If  {$x_i \notin \mathcal{S}$}
\State Increment distance
\EndIf       
\EndFor
\State \textbf{return} distance
\EndFunction
\end{algorithmic}
\end{varwidth}
\end{algorithm}

\emph{Remark 2 (probabilistic feedback):} In the case of severe constrained feedback, the receiver adds the mechanism of probabilistic feedback control, in which feedbacks are only transmitted with a given probability. An optimal feedback probability can be determined according to the capacity of back channel and the cost of feedback transmission. For instance, Fig. \ref{prob-feedback-overhead} shows simulation results of the coding overhead (i.e., number of forward transmissions normalized with respect to the number of input symbols) as the probability of sending $0$ and $1$ feedbacks increases. The results illustrate that when the probability of sending $0$ and $1$ distance feedbacks increases, amount of forward communications decreases. On the other hand, as Fig. \ref{prob-feedback-amount} shows, the (normalized) number of transmitted feedback messages increases with the probability, as expected. Therefore, by adjusting the probability of feedback transmission, decoder would be able to control the number of forward and feedback transmissions.  

   \begin{figure}[t!]
  \centering
  \subfigure[Number of forward transmissions]{\includegraphics[scale=0.33, trim=1cm .1cm 1cm .75cm, clip]{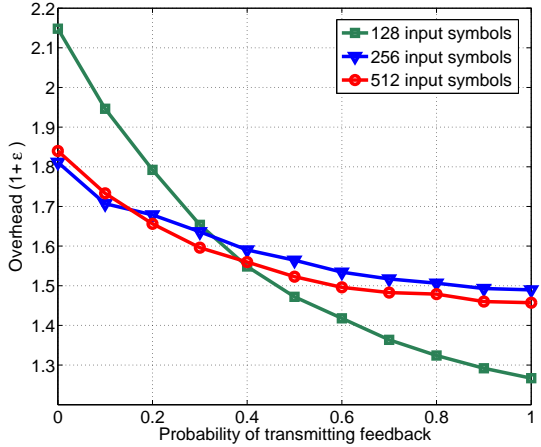}
  \label{prob-feedback-overhead}
  } \hspace{1cm}
  \subfigure[Number of feedbacks]{\includegraphics[scale=0.33, trim=1cm .1cm 1cm .75cm, clip]{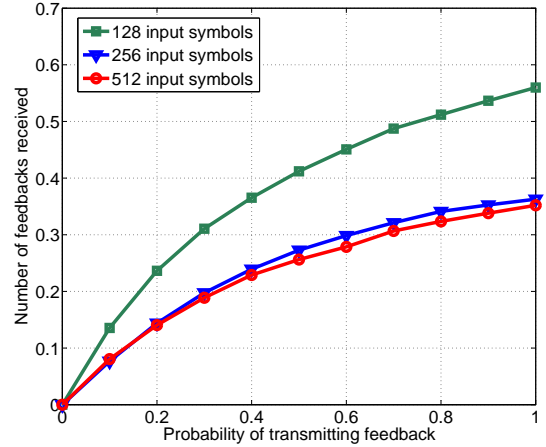}
  \label{prob-feedback-amount}
  }
 \caption[Probabilistic feedback control]{\subref{prob-feedback-overhead} Number of forward transmissions (normalized to the number of input symbols) needed by Delete-and-Conquer codes as the probability of sending feedback increases \subref{prob-feedback-amount} Number of feedback (normalized to the number of input symbols) as the probability of sending feedback increases.}
  \label{prob-feedback}
\end{figure}

\bb{Note that the Delete-and-Conquer encoder learns about the recovered symbols using a light-weight feedback and excludes the recovered symbols from subsequent transmissions. Alternatively, the receiver can send the identity of recovered symbols back to the transmitter. In this case, however, total amount of feedback  (up to $k \log(k)$ bits) is larger than the Delete-and-Conquer scheme. In fact, Fig. \ref{prob-feedback-amount} experimentally shows that total amount of feedback sent by the Delete-and-Conquer decoder is strictly less than $k$ bits. }  
 
\emph{Remark 3 (broadcast scenario):} To generalize the Delete-and-Conquer codes to the broadcast scenarios, we note that excluding a subset of recovered symbols from subsequent transmissions may increase the total number of transmissions (compared with when all recovered symbols are excluded), but it does not impede the decoding progress. In the worst case, no symbol is dropped from the encoding set, which reduces our codes to the original LT codes. Therefore, in a broadcast scenario, the encoder can simply take the intersection of collected feedbacks from different receivers, and  proceed with excluding those symbols confirmed to be recovered by all receivers.

%

\section{Short Block Length Analysis}
In this section, we precisely analyze the performance gains of the primitive form of our codes for very short block length of $k=2$ and $k=3$ symbols.  Although such small block lengths are not practical, they provide some insight into the Delete-and-Conquer scheme. For larger block lengths, our exact calculation of overhead in terms of degree probabilities becomes unwieldy. For the analysis purposes, we assume that the forward channel is lossless. 
\label{overhead}
\subsection{Block length k=2}  
As the first case, we consider the block length of $k=2$ symbols, in which two input symbols $x_1$ and $x_2$ are encoded. We assume that the probability of degree 1 transmission is equal to $2p$, and the probability of degree 2 transmission is $1-2p$. Therefore, an encoded symbol is equal to $x_1$ or $x_2$ each with probability $p$, and $x_1 + x_2$ with probability $1-2p$.

\begin{lemma}
\textit{For the block length $k=2$, if the probability of degree 1 transmission is $2p$, then the Delete-and-Conquer codes require an expected number of $\frac{4p^2+1}{2p}$ forward transmissions and $2p$ feedback transmissions for successful decoding.}
\label{savings}
\end{lemma}

\begin{proof}
A Delete-and-Conquer decoder can successfully decode two symbols within $n=2$ transmissions under the following possibilities for the received symbols:
$$
\begin{aligned}
\{x_1,x_2\}, \{x_2,x_1\}, \{x_1 + x_2, x_1\}, \{x_1 + x_2,x_2\}.     
\end{aligned}
$$ 
The probability of terminating after two transmissions is obtained as $4p-4p^2$. Similarly, the decoder would successfully recover $x_1$ and $x_2$ within $n \geq 3$ transmissions in the case of the following received symbols:
$$
\begin{aligned}
\{\overbrace{x_1 + x_2, ..., x_1 + x_2}^{n-1 \ \text{symbols}}, x_1\}, \ \{\overbrace{x_1 + x_2, ..., x_1 + x_2}^{n-1 \ \text{symbols}}, x_2\}. 
\end{aligned}
$$ 
The probability of successful recovery in  this case would be:
$$
Q(n)= (1-2p)^{n-1}\left(p+p\right), \ \ \ n\geq 3; 
$$
and therefore, the expected number of forward transmissions for the Delete-and-Conquer scheme is equal to: 
\begin{equation}
\bar{n}_{Del}=2(4p-4p^2)+ \sum_{n=3}^{\infty} n Q(n) = \frac{4p^2+1}{2p}.
\label{delete}
\end{equation}

To calculate the expected number of feedbacks transmitted, we note that one feedback is transmitted only in the cases of received symbols $\{x_1,x_2\}$ and $\{x_2,x_1\}$ each happens with probability $p$, and thus the expected number of feedbacks transmitted would be $2p$. It should be noted that the last feedback message is excluded from the count, as it is also needed by other coding schemes to stop the encoder from further transmissions.
\end{proof}

\begin{theorem}
\textit{For the block length $k=2$, the Delete-and-Conquer codes provide a savings of $\frac{2p^2}{1-p}$ in forward transmissions compared with the LT codes.}
\label{savings}
\end{theorem}

\begin{proof}
First, we calculate the expected number of transmissions required by the LT codes to recover all symbols. To this end, we obtain the probability of full recovery within $n \geq 2$ transmissions. For instance, in the case of $n=2$, the decoder should receive one of the following combinations to successfully recover $x_1$ and $x_2$: 
\begin{align}
\{x_1,x_2\}, \{x_1,x_1 + x_2\},  \{x_2,  x_1\}, \{x_2,x_1 + x_2\}, \{x_1 + x_2,x_1\}, \{x_1 + x_2,x_2\}.    \nonumber
\end{align}
Accordingly, the probability of decoding within two transmissions can be calculated as $4p-6p^2$. For a general case of $n$ transmissions, one can see that the probability of recovery within  $n$ transmissions is:
$$
P(n)=2p^{n-1} \left(p +(1-2p)\right) + \left(1-2p\right)^{n-1} \left( p + p \right);
$$
and hence, the expected total number of transmissions is: \begin{equation}
\bar{n}_{LT}= \sum_{n=2}^{\infty} n  P(n) = \frac{4p^2-p+1}{2p(1-p)}.
\label{LT}
\end{equation} 
Using \eqref{delete} and \eqref{LT},  expected amount of savings $\bar{n}_{LT} - \bar{n}_{Del}$  is obtained. 

\end{proof}

\begin{figure}[t!]
\centering
\includegraphics[scale=.50, trim=.4cm 2.4cm 1cm .75cm, clip]{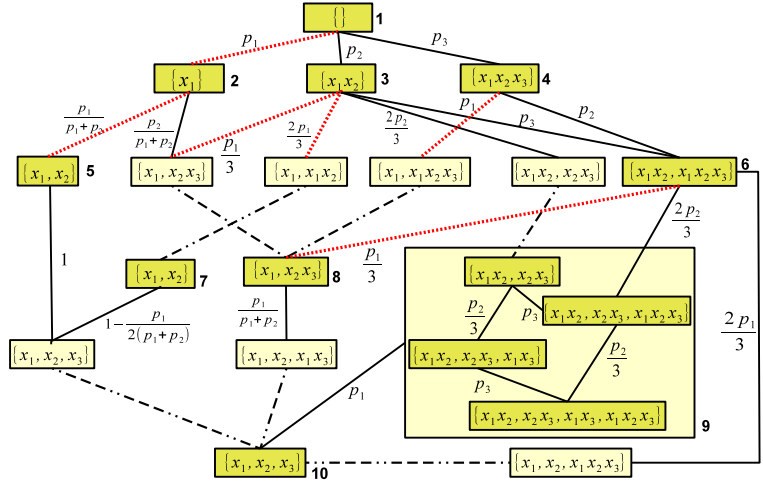}
\caption{State space of the Delete-and-Conquer scheme with $3$ input symbols.  The four states inside the box are considered as a single state. Notation $x_ix_j$ represents the symbol $x_i + x_j$, and dotted red lines represent transitions with a feedback.}
\label{fig:MarkovChain}
\end{figure}

\subsection{Block length k=3}
For the block length $k=3$, the authors in \cite{hyytia2007optimal} have derived the expected number of encoding symbols required by the LT codes for full recovery. In this model, the set of received symbols at the decoder defines a state of an absorbing Markov chain, and the process (i.e., transmission of encoded symbols) ends when it reaches to the absorbing state that includes all input symbols decoded. We similarly adapt this approach to obtain the Markov chain for the Delete-and-Conquer scheme with $3$ input symbols. The corresponding Markov chain shown in Fig. \ref{fig:MarkovChain}, includes states up to the permutations of input symbols, e.g., two states $\{x_1, x_2 + x_3\}$ and $\{x_2, x_1 + x_3\}$ are isomorphic and it is enough to consider a single unique state for each group of isomorphic states. \bb{In this figure, darker states are irreducible by the decoder in that no symbol can be further recovered, whereas other states can be immediately reduced by the decoder to the darker ones.}. By constructing the state transition matrix $\mathbf{P}$ as
 $$
 \mathbf{P} = \begin{pmatrix}
   \mathbf{Q}&\mathbf{R}   \\
   \mathbf{0}& \mathbf{I}\\
 \end{pmatrix},
 $$
we can compute the expected number of steps (transmissions) from the initial state to the absorbing state $\{x_1, x_2, x_3\}$. In the notation, matrix $\mathbf{Q}$ represents the transition probabilities between transient states, $\mathbf{R}$ denotes the probabilities between transient states and the absorbing state, and  $\mathbf{I}$ is an identity matrix.

\begin{theorem}
\textit{For the block length $k=3$, given that $p_j$ is the probability of transmitting an encoded symbol with degree $j$, the expected number of transmissions required by the Delete-and-Conquer scheme for successful decoding is:
\begin{align}
 \hspace*{-0.05cm}
\bar{n}_{Del} =    \frac{1}{p_1} +  \frac{p_2}{3p_1+2p_2}+ \frac{p_2^2}{p_1+p_2} - \frac{8p_2^3}{(p_1+2p_2)(p_2-3)} + \frac{3p_1-4p_2+3p_1p_2-3p_2^3+3}{3-p_2} .
\label{delete-num-transmit}
\end{align}}
\end{theorem}

 \begin{proof}
 In an absorbing Markov chain with a transition matrix $\mathbf{P}$ and the \emph{fundamental matrix} 
$$\mathbf{N = I+ Q+ Q^2 + ... = (I-Q)^{-1}},$$
the expected number of steps (transmissions) from the initial state to the absorbing one is:
\begin{equation}
\bar{n} = \bm \pi_{0} \mathbf{N} \mathbf{c},
\label{expected-steps}
\end{equation}
where $\bm \pi_{0} = (1 \ 0 \ ...\ 0)$ is the initial probability corresponding to the state of no symbols been transmitted, and $\mathbf{c}=(1 \ .. \ 1)^{T}$ \cite{ross2006introduction}. From Fig. \ref{fig:MarkovChain}, we obtain matrix $\mathbf{P}$ as:
\footnotesize
$$
\hspace{-.35cm}
\mathbf{P} = \begin{pmatrix}
  0 & 	p_1 	& 	p_2	 & 	p_3 	&	0 		& 	0	&	 0	&	0	 	&	0	&	0 \\
  0 & 	0 		&	0		& 		0		 & p_1'	& 	0	&	 0 &	p_2' 	&	0	&	0 \\
  0 & 0	& \frac{p_2}{3}	& 	0 & 	0	 & 	p3	& \frac{2p_1}{3} 	& \frac{p_1}{3} & \frac{2p_2}{3} &0\\
  0 &	0 	&	0 				&p_3 &0 & p_2& 0	&	p_1	&	0 &	0 \\
  0 & 0&0 & 0& 0& 0& 0&0&0&1 \\
  0 & 0 & 0& 0 & 0 & \frac{p_2+3p_3}{3} & 0 & \frac{p_1}{3} & \frac{2p_2}{3} & \frac{2p_1}{3} \\
  0 & 0 & 0 &0 &0 & 0 & \frac{p_1'}{2} & 0 & 0 & 1-\frac{p_1'}{2} \\
   0& 0& 0&0 &0 &0 &0 &1-p_1'&0&p_1' \\
	0& 0& 0&0 &0 &0 &0 &0&1-p_1&p_1 \\
  0 & 0& 0& 0& 0& 0& 0&0&0&1 
 \end{pmatrix}
$$
\normalsize
where we assume that after each symbol deletion at the encoder, the probabilities are normalized by dividing by the sum of the remaining degrees. For instance, after one exclusion, $p_1' \triangleq \frac{p_1}{p_1+p_2}$ and $p_2'\triangleq\frac{p_2}{p_1+p_2}$ would be the probability of degrees $1$ and $2$ transmissions respectively. This leads to the theorem statement.
\end{proof}

The expected number of transmissions for the LT codes has been derived in \cite{hyytia2007optimal} as follows:
\small
\begin{align}
 \hspace*{-0.05cm}
\bar{n}_{LT} = \dfrac{1}{p_1} + \dfrac{6p_1}{p_1-3}+ \dfrac{18p_1}{(3-p_2)(3-2p_1-p_2)}+\dfrac{9p_1}{2(p_1+p_2)(3p_1+2p_2)}.
\label{lt-num-transmit}
\end{align}
\normalsize
If the encoder uses only degree $1$ symbols (i.e., $p_1=1$), the expected number of required symbols for the LT codes is $\bar{n}_{LT}=5.5$, illustrating the effect of the coupon collector's problem; on the other hand, Delete-and-Conquer scheme requires only $\bar{n}_{Del}=3$ encoded symbols, which is the minimum possible number of forward transmissions. It should be noted that in this case, Delete-and-Conquer scheme turns into a no-coding ARQ method. An optimization in~\cite{hyytia2007optimal} results in a minimum number  of $4.046$ forward transmissions (with $p_1=0.524, p_2=0.366$, and $p_3=0.109$) for the LT codes, whereas Delete-and-Conquer coding with these same probabilities yields a total number of $\bar{n}_{Del}=3.678$ forward transmissions.  In general, we can numerically compare~\eqref{delete-num-transmit} to~\eqref{lt-num-transmit} to see that Delete-and-Conquer scheme can decrease the total number of forward transmissions up to $2.4$-fold. 

\begin{theorem}
\textit{For $k=3$ input symbols, the expected number of feedbacks transmitted by the Delete-and-Conquer scheme before conclusion (i.e. not including the termination signal) is:
\begin{equation}
\bar{f}_{Del} = \frac{3p_1}{3p_1+2p_2} + \frac{6p_1}{3-p_2}+\frac{p_1^2}{p_1+p_2} -2p_1.
\label{expected_feedback}
\end{equation}}
\end{theorem}

\begin{proof}
In an absorbing Markov chain, the probability of ever visiting state $j$ when starting at a transient state $i$ is the entry $h_{ij}$ of the matrix $\mathbf{H=(N-I)N_{dg}^{-1}}$, where $\mathbf{N}$ is the fundamental matrix and $\mathbf{N_{dg}}$ is the diagonal matrix with the same diagonal as $\mathbf{N}$, and $\mathbf{I}$ is an identity matrix~\cite{ross2006introduction}. In Fig. \ref{fig:MarkovChain}, a feedback is transmitted when transitions along the dotted-line occur, e.g. a transition from the state $1$ to state $2$. Accordingly, the probability of such transitions, and hence the expected number of feedbacks transmitted is given by:
\begin{align}
\bar{f}_{Del} = h_{12} + h_{12} h_{25} + h_{13} h_{37} + h_{13} h_{38} + h_{14} h_{48}; \nonumber
\end{align}
\normalsize
from which the result follows.
\end{proof}

Based on the Theorem 2 and 3, we can calculate the optimal probability values $p_1^*$ and $p_2^*$ (and $p_3^*=1-p_1^*-p_2^*$) that minimize the total number of forward and feedback transmissions needed by the Delete-and-Conquer scheme. In other words:

$$
(p_1^*, p_2^*) = \argmin_{(p_1,p_2)} \left[\bar{n}_{Del} + \bar{f}_{Del} \right];
$$
which results in $(p_1^*, p_2^*) = (0.644, 0.206)$ (and $p_3^*=0.150$) with a minimum number of total transmissions $4.7247$. In this case, we simply considered the sum of forward and feedback transmissions. In a more general sense, we can assume that each transmission through the forward channel has a cost of $C_1$, while each feedback transmission has a cost of $C_2$. Therefore, the optimal probability values can be calculated as:
$$
(p_1^*, p_2^*) = \argmin_{(p_1,p_2)} \left[C_1\bar{n}_{Del} + C_2\bar{f}_{Del} \right].
$$   
Moreover, in comparison with the LT codes, one can notice that it is worthwhile to send feedback if:
$$
  C_1 \bar{n}_{Del} + C_2 \bar{f}_{Del} \leq C_1 \bar{n}_{LT} \Rightarrow \frac{\bar{f}_{Del}}{\bar{n}_{LT}-\bar{n}_{Del}} \leq \frac{C_1}{C_2}, 
$$ 
where $\bar{n}_{Del}$, $\bar{n}_{LT}$, and $\bar{f}_{Del}$ are calculated in (\ref{delete-num-transmit}), (\ref{lt-num-transmit}), and (\ref{expected_feedback}). 

\section{Maximum-Likelihood Decoder Analysis}
\label{maximum-likelihood}
In this section, we derive an upper-bound on the failure probability of the maximum likelihood (ML) decoder when used with the Delete-and-Conquer codes. We assume that there are $k$ input symbols at the transmitter, and that $n$ encoding symbols are received over a binary erasure channel (BEC). The ML decoding  over a BEC is equivalent to recovering $k$ information (input) symbols  from $n$ received encoding (output) symbols. Without loss of generality, we assume that each symbol is one bit; $\mathbf{x}$ is a row vector containing $k$ input bits; and $\mathbf{y}$ is the vector of $n$ output bits. Matrix $\mathbf{G}=[g_{i,j}]$ is an $n\times k$ adjacency matrix of the decoder graph, such that an entry $g_{i,j}$ is equal to $1$ if the $i^{th}$ output node has the $j^{th}$ input node as a neighbor. The ML decoder is then equivalent to solving a system of linear equations (with unknowns $\mathbf{x}$ and received symbols $\mathbf{y}$) of the form:

\begin{equation}
\mathbf{Gx}^T = \mathbf{y}^T. 
\label{ML-decoder}
\end{equation} 

%
%

Encoding symbols with a distance of $0$ or $1$ trigger a feedback message that causes the corresponding symbols to be excluded from future transmissions. Excluding the recovered symbols from subsequent transmissions is equivalent to setting the subsequent elements of the corresponding columns in $\mathbf{G}$ to zero. For instance, Fig. \ref{feedback-matrix} shows a realization of the matrix $\mathbf{G}$ in which the first feedback message acknowledges recovery of $x_2$. Thereafter the second column of $\mathbf{G}$ (i.e., the shaded part) is set to zero. 

\begin{figure}[h]
\centering
\includegraphics[scale=.43, trim=1.5cm 6.9cm 2.8cm 2.7cm, clip]{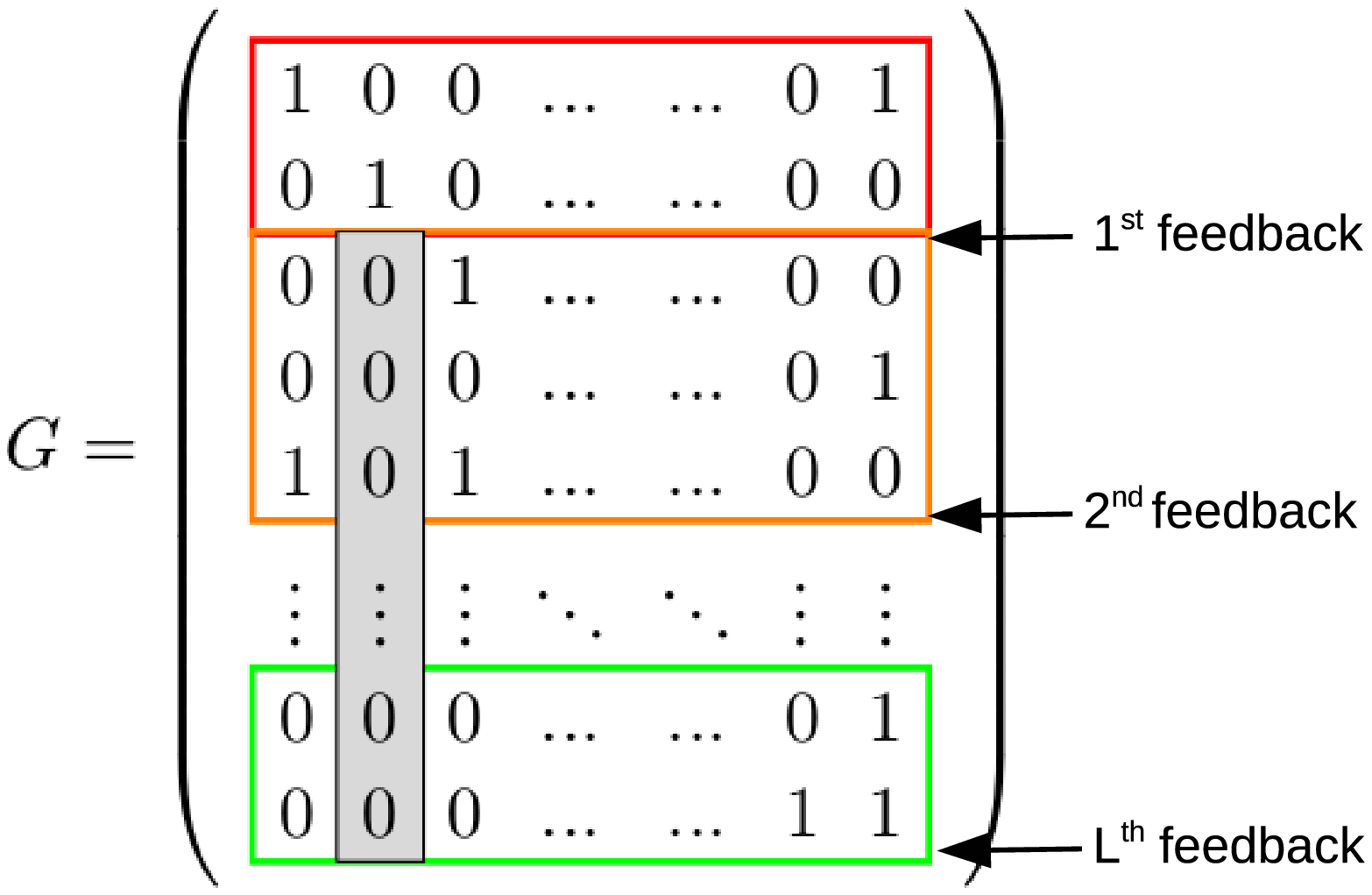}
\caption{Decoder matrix and the location of feedback transmission}
\label{feedback-matrix}
\end{figure}
 
The ML decoder failure is equivalent to the event that the adjacency matrix $\mathbf{G}$ in \eqref{ML-decoder} is not of full rank. Let $p_e$ be the probability that an input bit $j$ (for an arbitrary $j \in \{1,2,..k\}$) is not recoverable under the ML decoding rule. From \cite{rahnavard2007rateless}:

\begin{align}
p_e =  \text{Pr} \left\lbrace \exists \mathbf{x} \in GF(2^k), x_j=1 : \ \  \mathbf{G}\mathbf{x}^T= \mathbf{0}^T \right\rbrace  \leq \sum_{\substack{\mathbf{x} \in GF(2^k) \\ x_j =1}} \text{Pr} \left\lbrace \mathbf{G}\mathbf{x}^T= \mathbf{0}^T \right\rbrace. 
\label{error2}
\end{align}

In order to calculate $\text{Pr} \{ \mathbf{Gx}^T= \mathbf{0}^T \}$, we separately consider the rows of  $\mathbf{G}$ between consecutive feedback messages.  We assume that $L$ feedback messages are transmitted in total such that after receiving $t_1$ encoding symbols the first feedback message is transmitted, after receiving $t_2$ encoding symbols the second feedback is sent, and so forth. At the boundary points, we define $t_0=0$ and $t_{L}=n$. Therefore, there is no feedback within each interval of [0, $t_1$], ($t_1$, $t_2$], ..., ($t_{L-1}$,~n], and there is one feedback at the end of each interval, as shown in Fig. \ref{feedback-interval}. We assume that within the $i^{th}$ interval ($i=0, ..., L-1$) the coding window contains $k-m_i$ symbols, and thus the encoder uses a fixed degree distribution $\Omega_{k-m_i}(d)$ defined over the set of $k-m_i$ unacknowledged symbols.

\begin{figure}
\centering
\includegraphics[scale=.4, trim=.5cm 1cm 1cm 1cm, clip]{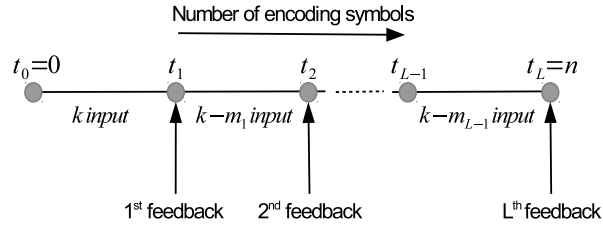}
\caption{Location of feedbacks and interval of coding}
\label{feedback-interval}
\end{figure}

 To calculate an upper bound on the decoder failure probability, we start with a single row of $\mathbf{G}$. Let $\mathbf{r}$ to be a row of degree $d$, and assume that the total number of $m$ symbols are acknowledged before transmitting $\mathbf{r}$; in other words, $m$ indices out of $k$ indices in $\mathbf{r}$ are forced to be zero. We define a row vector $\mathbf{f}$ such that $f_l =1$ if the $l^{th}$ symbol has been acknowledged, and $0$ otherwise (i.e., an indicator function on the index of acknowledged symbols). For a given input vector $\mathbf{x}$ with $||\mathbf{x}||_0 = w$ ($||.||_0$ is the 0-norm), we have the following lemma:
 \begin{lemma}
 Given that the row vector $\mathbf{r}$ has degree $d$ (i.e., $||\mathbf{r}||_0 = d$), the probability of $\mathbf{r}\mathbf{x}^T = 0$ is:
 $$
 p \bigl( \mathbf{x}, ||\mathbf{r}||_0 = d \bigr) =  \frac{\sum_{\substack{u=0,2,.., \min (2\lfloor \frac{d}{2}\rfloor, \bar{w}) }} \binom{\bar{w}}{u} \binom{k-m-\bar{w}}{d-u} }{\binom{k-m}{d}},
 $$
 \label{lemma4}
 \end{lemma} 
in which $\bar{w} = w- \langle \mathbf{x},\mathbf{f} \rangle$ with $\langle \mathbf{x},\mathbf{f} \rangle$ denoting the dot product of two vectors. 
\begin{proof}
The event $\mathbf{rx}^T = 0$ happens if and only if $\mathbf{r}$ has an even number of $1$'s in those indices of $j$ in which $x_j$ is equal to $1$ as well. Assume that $J=\{j_1, j_2, .., j_w\}$ is the set of indices in which $\mathbf{x}$ is $1$, and $A=\{a_1, a_2, ..., a_m\}$ is the set of acknowledged indices. Therefore, we need to choose an even number $u$ of indices that belong to $J$ but not to $A$. The number of these non-overlapping indices is given by $\bar{w} = w - \langle \mathbf{x},\mathbf{f} \rangle$. Because the degree of $\mathbf{r}$ is $d$,  we then need to choose $d-u$ symbols from the remaining $k-m- w + \langle \mathbf{x}, \mathbf{f} \rangle$ indices that belong neither to $J$ nor to $A$. Finally, given that the vector $\mathbf{r}$ is generated randomly (i.e., $d$ neighbors are selected uniformly at random from $k-m$ unacknowledged symbols), the result follows.
\end{proof}

Using Lemma \ref{lemma4} and the fact that $\mathbf{r}$ has degree $d$ with probability $\Omega_{k-m}(d)$, we have:  
\begin{equation}
p(\mathbf{x}) = \sum_{d=1}^{k-m} \Omega_{k-m}(d) p \bigl( \mathbf{x}, ||\mathbf{r}||_0 = d \bigr).
\label{probability}
\end{equation}
Now, we can extend this result to more than one row of $\mathbf{G}$ in the following manner. Let us denote the $t_{i} - t_{i+1}$  rows of $\mathbf{G}$ by $\mathbf{G}_i$.  
Rows in $\mathbf{G}_i$ are generated independently according to the degree distribution $\Omega_{k-m_i}(d)$. Thus, we have: 
\begin{align}
 \text{Pr} \{\mathbf{G}_i\mathbf{x}^T= \mathbf{0}^T \} = \left(p_i(\mathbf{x})\right)^{t_{i+1}-t_{i}},
\end{align}   
in which $p_i(\mathbf{x})$ is calculated as in \eqref{probability}, and based on the number of acknowledged symbols and the degree distribution within the $i^{th}$ interval, i.e.,: 
$$ p_i(\mathbf{x}) = \sum_{d=1}^{k-m_i} \Omega_{k-m_i}(d) p_i \bigl( \mathbf{x}, ||\mathbf{r}||_0 = d \bigr).$$
Given that there are $L$ transmit intervals (i.e., $L$ feedback messages), we can calculate the probability of $\mathbf{G x}^T= \mathbf{0}^T$ for a given vector $\mathbf{x}$ as follows: 
\begin{equation}
 \text{Pr} \{\mathbf{G x}^T= \mathbf{0}^T \} = \prod_{i=0}^{L-1} \text{Pr} \{\mathbf{G}_i \mathbf{x}^T= \mathbf{0}^T \}. 
 \label{vector}
\end{equation}
Assembling these steps together, the ML decoder failure probability of the Delete-and-Conquer scheme is given by the following theorem.   
\bb{\begin{theorem}
Given that $L$ feedbacks are transmitted in total (i.e, one feedback after receiving the $t_i$-th ($i=1,...,L$) encoding symbol), the ML decoder failure probability of recovering an input symbol $j$ (for an arbitrary $j \in \{1,2,..k\}$) is upper bounded by
\begin{align}
p_e  \leq  \min \bigg\{1, \sum_{w=1}^{k} \big( \sum_{\substack{\mathbf{x} \\ ||\mathbf{x}||_0 =w \\ x_j=1}}  \prod_{i=0}^{L-1} \text{Pr} \{\mathbf{G}_i \mathbf{x}^T= \mathbf{0}^T \}\big) \bigg\}.
\label{BACK-bound}
\end{align}
\end{theorem}}

\begin{proof}
From Lemma \ref{lemma4} and its following results, we obtain that for a given input vector $\mathbf{x}$ with Hamming weight $w$ and $L$ feedback messages, the probability of  $\mathbf{G x}^T= \mathbf{0}^T $ is calculated as in Eq. \eqref{vector}. Therefore, summing over all possible input vectors $\mathbf{x}$ with the $j^{th}$ index equal to $1$, yields the theorem statement. It should be noted that we assume the values of $L$  and $t_i$'s (i.e., the total number of feedbacks and their trigger points) are known.  
\end{proof}

From \cite{rahnavard2007rateless}, the upper bound on the ML decoder failure probability when there is no feedback is calculated as:
\begin{align}
p_e  \leq \min \Bigg\{ 1,  &\sum_{w=1}^{k} \binom{k-1}{w-1}  \left(\sum_d \Omega_k(d) \frac{\sum_{\substack{s=0,2,..,2\lfloor \frac{d}{2} \rfloor} } \binom{w}{s} \binom{k-w}{d-s} }{\binom{k}{d}} \right)^n \Bigg\}.
\label{LT-bound}
\end{align} 
Fig. \ref{bound-comparison} numerically compares the upper bound in~\eqref{BACK-bound} with that in \eqref{LT-bound} for $k=100$ input symbols. The results confirm that collecting more encoding symbols reduces the bound on ML failure probability, as expected. However, Delete-and-Conquer codes with $0$ and $1$ feedback messages achieve a tighter upper-bound on the decoder failure probability.     

\begin{figure}[t!]
\centering
\includegraphics[scale=.4, trim=.5cm 0cm 1cm 1cm, clip]{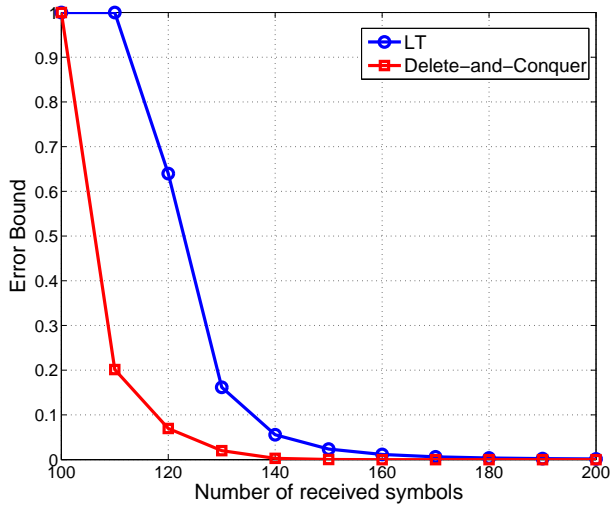}
\caption{Upper bound on the maximum likelihood decoder failure probability for the LT and the Delete-and-Conquer codes}
 \label{bound-comparison}
\end{figure}

\paragraph*{Asymptotic results}
We conclude our analysis of the Delete-and-Conquer scheme by providing upper bounds on its performance metrics. To this end, we keep the assumption that initially there exist $k$ input symbols at the encoder, and at some point, $m$ symbols are acknowledged. 
The Delete-and-Conquer distribution is thus given by $\Omega_{k-m}(i)$  (for $i=1, ..., k-m$).
Adapting the results of~\cite{luby2002lt} yields that the average degree of an encoding symbol generated by the encoder is given by $\bar{D}=O(\ln (k-m))$. Furthermore, an encoder that deletes $m$ symbols out of $k$ symbols, needs to transmit \emph{at most}
$
k-m + O\left(\sqrt{k-m} \ \ln^2(\frac{k-m}{\delta})\right)
$
encodings so that the decoder be able to recover all input symbols with probability at least $1-\delta$. Furthermore, computational complexity of the coding process is given by $O \left((k-m) \ln \frac{k-m}{\delta}\right)$. 
 One can notice that as the number of acknowledged symbols (i.e., parameter $m$) increases, performance metrics improves. In the case of no feedbacks (i.e., $m=0$) Delete-and-conquer codes reduce to the original LT codes.

\section{Simulation Results}
\label{simulation-results}
We evaluate the performance of rateless codes with nonuniform selection distributions against the Growth codes, Online codes proposed in \cite{cassuto2011line}, and recently proposed LT-AF codes \cite{talari2014robust}.

\subsection{General form}
\subsubsection{Intermediate performance}
In many applications such as video streaming with real-time playback requirements, it is essential to partially recover some symbols before the recovery of entire frame. In this context, although LT codes are capacity-achieving, they lack real-time features; in other words, not many input symbols are decoded until the decoding process is almost complete. By incorporating a nonuniform selection distribution at the encoder, we aim to enhance the intermediate symbol recovery rate. Fig. \ref{intermediate2} compares the performance of our codes with the LT-AF codes of Variable Node with Maximum Degree (LT-AF+VMD) \cite{talari2014robust}, where the authors show that LT-AF codes can surpass previous rateless codes with feedback including SLT codes. One key point, however, is that the LT-AF decoder is not able to recover any symbol until at least $k$ encoding symbols are received. As the results show, our scheme based on the Quantized distance method can achieve a high intermediate recovery rate. Moreover, the coding performance can be adjusted by tuning the parameter $s$ (feedback transmission interval).

\subsubsection{Coding overhead}
Next we compare the total number of forward and feedback transmissions needed by our codes in comparison with the LT-AF+VMD codes. As the results in Table \ref{tab:LT-AF} show, our codes have a slightly better  performance in terms of number of forward transmissions. However, LT-AF codes require less feedback transmissions. It should be noted that  amount of feedback in our codes can be adjusted using the parameter $s$, and that our codes are aimed to achieve a high intermediate symbol recovery rate, as the results in Fig. \ref{intermediate2} show.  
\begin{figure}[t]
\centering
\includegraphics[scale=0.4, trim=.5cm 2cm 1cm 1cm, clip]{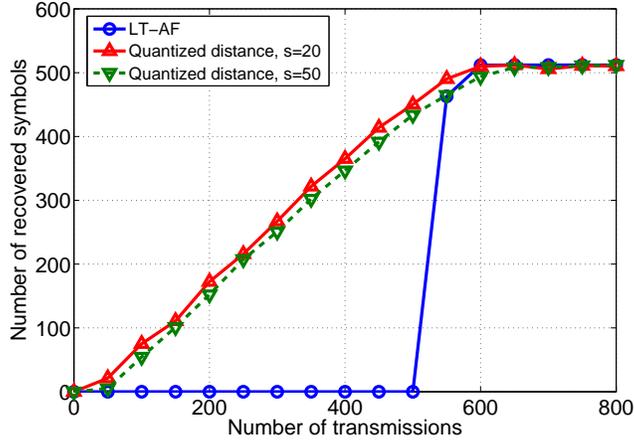}
  \caption[Optional caption for list of figures]{Intermediate performance of codes with nonuniform symbol selection against the LT-AF codes ($k=512$).}
 \label{intermediate2}
\end{figure}
\begin{table}[t]\centering
\renewcommand{\arraystretch}{1.3}
\ra{1.3}
\begin{tabular}{@{}lcccc@{}}\toprule
\multirow{2}{*}{Algorithm} 
& \multicolumn{2}{c}{k=512}  & \multicolumn{2}{c}{k=1024}\\
 \cmidrule{2-3} \cmidrule{4-5}
  & Forward & Feedback & Forward & Feedback  \\ \toprule
 LT-AF + VMD & 556.0 & 9.0 & 1084.0 & 11.8 \\ \midrule
 All-Distance  & 550.4 & 54.4 & 1084.8 & 107.8 \\ \midrule
 Quantized distance & 555.2 & 55.0 & 1112.6 & 111.0 \\
 \bottomrule
\end{tabular}
\caption{Number of transmissions needed by the LT-AF codes and our codes with $s=10$}
  \label{tab:LT-AF}
\end{table}
Table \ref{tab:parameter-s} shows the performance of our codes with the block length $k=512$ symbols and as the feedback interval $s$ increases. Similar to the previous results, the encoder is able to control the number of forward and feedback transmissions by changing the parameter $s$.  

\subsection{Primitive form}
\subsubsection{Intermediate performance}
 To investigate the progressive performance of the Delete-and-Conquer codes, we run simulations with the block length of $k=512$.   Results shown in Fig. \ref{intermediate1}, demonstrate that Growth codes can provide higher symbol recovery rate at the beginning, while Delete-and-Conquer achieves better performance when a small fraction of symbols are unrecovered (near the ``knee''). On the other hand, Delete-and-Conquer scheme achieves better performance compared with the Online codes, noting that Delete-and-Conquer codes improve the intermediate performance with a lightweight utilization of the back channel (i.e., one bit feedback for each of a small fraction of received symbols).  
 \begin{table}[]\centering
\renewcommand{\arraystretch}{1.3}
\ra{1.3}
\begin{tabular}{@{}lcccc@{}}\toprule
\multirow{2}{*}{Feedback interval} 
& \multicolumn{2}{c}{All-Distance}  & \multicolumn{2}{c}{Quantized distance}\\
 \cmidrule{2-3} \cmidrule{4-5}
  & Forward & Feedback & Forward & Feedback  \\ \toprule
 $s=5$ & 536.6 & 106.9 & 544.8 & 108.4 \\ \midrule
 $s=10$  & 550.4 & 54.4 & 555.2 & 55.0 \\ \midrule
 $s=50$ & 657.6 & 12.8 & 684.8 & 13.6\\ \midrule
 $s=100$ & 758.6 & 7.0 & 759.4 & 7.2 \\ \midrule
 $s=500$ & 1155.7 & 2.0 & 1172.6 & 2.0\\  \bottomrule
\end{tabular}
\caption{Number of transmissions needed by our codes as the interval of feedback transmission increases ($k=512$)} \label{tab:parameter-s}
\end{table}

\begin{figure}[t]
\centering
\includegraphics[scale=.4, trim=.5cm .5cm 1cm 1cm, clip]{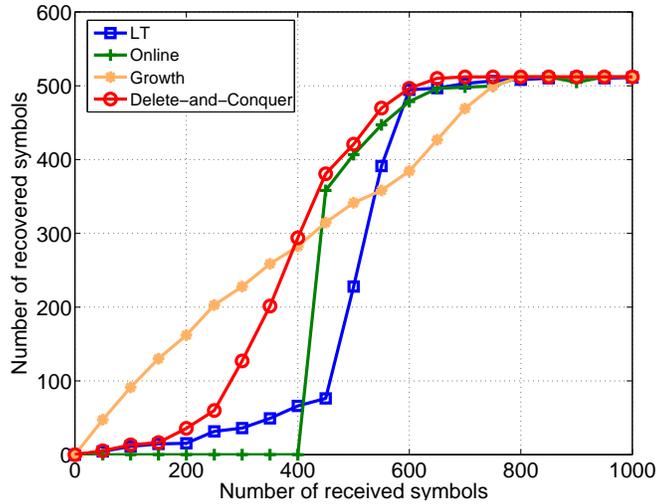}
  \caption[Optional caption for list of figures]{ Intermediate performance of the Delete-and-Conquer codes compared with other rateless codes ($k=512$)}
 \label{intermediate1}
\end{figure}

\subsubsection{Computational complexity}
Computational costs at the encoder and decoder are mainly related to the average degree of input symbols.
Fig.~\ref{fig:avg-degree} shows the average degree of input symbols for different codes compared to the Delete-and-Conquer codes. As the results show, Delete-and-Conquer codes have a smaller average degree on input symbols, and hence they incur less computational complexity. Smaller average degree is due to incrementally dropping input symbols from the coding window. 

\begin{figure}[h]
\centering
\addtolength{\subfigcapskip}{-0.05in}
\includegraphics[scale=.38, trim=1cm 0cm 1cm .75cm, clip]{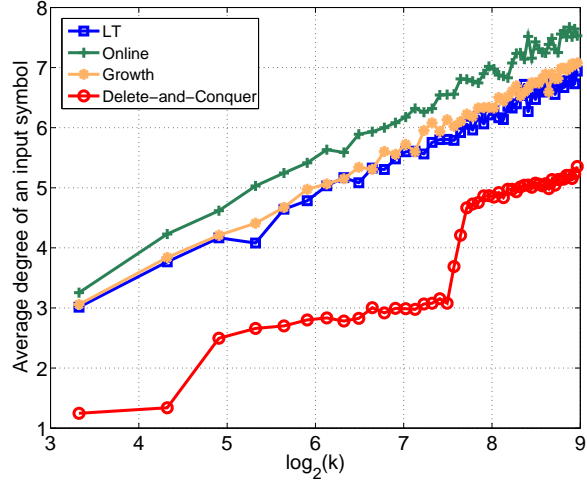}
\caption{Average degree of input symbols for various coding schemes }
\label{fig:avg-degree}
\end{figure}


\section{Conclusion}
\label{conclusion}
In this paper, we have developed feedback-based rateless codes with a nonuniform selection distribution.  Our encoders estimate the decoder state using feedback information, and dynamically adjust the selection distribution so that more helpful symbols (in terms of decoding progress) are assigned with a higher probability to be included in future encodings. As a result, we improve the intermediate performance of the underlying rateless codes and make them more suitable for applications with real-time decoding requirements. Our codes further support two important features: our decoder has full control of the \emph{rate} and \emph{timing} of feedback transmission. Our simulation results, backed by analysis, confirm that distance-type feedback paired with a nonuniform selection distribution achieves a high intermediate recovery rate.  On the whole, rateless codes with nonuniform selection distributions help the encoder to optimize for the performance requirements dictated by the application.    

\bibliographystyle{ieeetr} 
\bibliography{Summaries}

\begin{thebibliography}{10}

\bibitem{allerton2013}
M.~Hashemi, A.~Trachtenberg, and Y.~Cassuto, ``Delete-and-conquer: Rateless
  coding with constrained feedback,'' in {\em 51st Annual Allerton Conference},
  2013.

\bibitem{allerton2014}
M.~Hashemi and A.~Trachtenberg, ``Near real-time rateless coding with a
  constrained feedback budget,'' in {\em 52nd Annual Allerton Conference},
  2014.

\bibitem{luby2002lt}
M.~Luby, ``{LT} codes,'' in {\em Annual Symposium on Foundations of Computer
  Science}, pp.~271--280, 2002.

\bibitem{shokrollahi2006raptor}
A.~Shokrollahi, ``Raptor codes,'' {\em IEEE Transactions on Information
  Theory}, vol.~52, no.~6, pp.~2551--2567, 2006.

\bibitem{beimel2007rt}
A.~Beimel, S.~Dolev, and N.~Singer, ``{RT} oblivious erasure correcting,'' {\em
  IEEE/ACM Transactions on Networking}, vol.~15, no.~6, pp.~1321--1332, 2007.

\bibitem{hagedorn2009rateless}
A.~Hagedorn, S.~Agarwal, D.~Starobinski, and A.~Trachtenberg, ``Rateless coding
  with feedback,'' in {\em IEEE INFOCOM}, pp.~1791--1799, 2009.

\bibitem{byers2002digital}
J.~W. Byers, M.~Luby, and M.~Mitzenmacher, ``A digital fountain approach to
  asynchronous reliable multicast,'' {\em IEEE Journal on Selected Areas in
  Communications}, vol.~20, no.~8, pp.~1528--1540, 2002.

\bibitem{gallager1960low}
R.~G. Gallager, {\em Low density parity check codes}.
\newblock PhD thesis, Massachusetts Institute of Technology, 1960.

\bibitem{berrou1993near}
C.~Berrou, A.~Glavieux, and P.~Thitimajshima, ``Near shannon limit
  error-correcting coding and decoding: Turbo-codes,'' in {\em IEEE
  International Conference on Communications}, vol.~2, pp.~1064--1070, 1993.

\bibitem{byers1998digital}
J.~W. Byers, M.~Luby, M.~Mitzenmacher, and A.~Rege, ``A digital fountain
  approach to reliable distribution of bulk data,'' in {\em ACM SIGCOMM},
  vol.~28, pp.~56--67, 1998.

\bibitem{mackay2005fountain}
D.~J. MacKay, ``Fountain codes,'' in {\em IEE Proceedings-Communications},
  vol.~152, pp.~1062--1068, 2005.

\bibitem{sorensen2012rateless}
J.~H. S{\o}rensen, T.~Koike-Akino, and P.~Orlik, ``Rateless feedback codes,''
  in {\em IEEE International Symposium on Information Theory}, pp.~1767--1771,
  2012.

\bibitem{hyytia2007optimal}
E.~Hyytia, T.~Tirronen, and J.~Virtamo, ``Optimal degree distribution for {LT}
  codes with small message length,'' in {\em IEEE INFOCOM}, pp.~2576--2580,
  2007.

\bibitem{talari2009rateless}
A.~Talari and N.~Rahnavard, ``Rateless codes with optimum intermediate
  performance,'' in {\em IEEE Global Telecommunications Conference}, pp.~1--6,
  2009.

\bibitem{kamra2006growth}
A.~Kamra, V.~Misra, J.~Feldman, and D.~Rubenstein, ``Growth codes: Maximizing
  sensor network data persistence,'' in {\em ACM SIGCOMM}, vol.~36,
  pp.~255--266, 2006.

\bibitem{cassuto2011line}
Y.~Cassuto and A.~Shokrollahi, ``On-line fountain codes for semi-random loss
  channels,'' in {\em IEEE Information Theory Workshop (ITW)}, pp.~262--266,
  2011.

\bibitem{talari2014robust}
A.~Talari and N.~Rahnavard, ``Robust {LT} codes with alternating feedback,''
  {\em Computer Communications}, 2014.

\bibitem{ross2006introduction}
S.~M. Ross, {\em Introduction to probability models}.
\newblock Access Online via Elsevier, 2006.

\bibitem{rahnavard2007rateless}
N.~Rahnavard, B.~N. Vellambi, and F.~Fekri, ``Rateless codes with unequal error
  protection property,'' {\em IEEE Transactions on Information Theory},
  vol.~53, no.~4, pp.~1521--1532, 2007.

\end{thebibliography}

\end{document}